\documentclass[11pt]{article}
\usepackage{epsfig,amsmath,amsthm,latexsym,graphicx,amsfonts,amssymb,rotating}
\usepackage{rotating}
\newcommand{\beq}{\begin{equation}}
\newcommand{\eeq}{\end{equation}}

\newcommand{\R}{\mathbb R}

\newcommand{\E}{\mathbb E}
\newcommand{\BH}{\mathbb H}

\newcommand{\gammafn}{\operatorname{\Gamma}}
\newcommand{\grad}{{\nabla}}

\newcommand{\var}{\operatorname{\rm Var}}

\newcommand{\col}{\operatorname{\rm col}}
\newcommand{\uh}{\phi}
\newcommand{\uhspace}{\Phi}

\newcommand{\htheta}{\hat{\theta}}
\newcommand{\hq}{\hat{q}}
\newcommand{\horomap}{H}
\newcommand{\horomapx}{\horomap_\Xi}
\newcommand{\horomapu}{\horomap_\uhspace}
\newcommand{\fxt}{f_{\Xi\Theta}}
\newcommand{\fxu}{f_{\Xi\uhspace}}
\newtheorem{theorem}{Theorem}

\newtheorem{lemma}[theorem]{Lemma}

\newtheorem{defn}{Definition}
\newtheorem{remark}{Remark}
\def\defeq{\stackrel{\mbox{\tiny{def}}}{=}}
\begin{document}

\title{SMML estimators for linear regression \\ and tessellations of hyperbolic space}
\author{James G. Dowty}

\maketitle

\abstract{
The strict minimum message length (SMML) principle links data compression with inductive inference.  The corresponding estimators have many useful properties but they can be hard to calculate.  We investigate SMML estimators for linear regression models and we show that they have close connections to hyperbolic geometry.  When equipped with the Fisher information metric, the linear regression model with $p$ covariates and a sample size of $n$ becomes a Riemannian manifold, and we show that this is isometric to $(p+1)$-dimensional hyperbolic space $\mathbb{H}^{p+1}$ equipped with a metric tensor which is $2n$ times the usual metric tensor on $\mathbb{H}^{p+1}$.  A natural identification then allows us to also view the set of sufficient statistics for the linear regression model as a hyperbolic space.  We show that the partition of an SMML estimator corresponds to a tessellation of this hyperbolic space.
}

\section{The linear regression model}

To establish our notation we briefly recall some details of the linear regression model.

The linear regression model is a statistical model for observed data $y \in \R^n$ (thought of as a column matrix) which is a realization of an $n$-dimensional, normally-distributed random variable $Y$ with mean $A \beta$ and variance-covariance matrix $\sigma^2 I_n$, i.e.,
$$Y \sim N_n(A \beta, \sigma^2 I_n),$$
where $A$ is a full-rank $n \times p$ matrix called the design matrix, $\beta \in \R^p$ is a column matrix, $\sigma >0$ and $I_n$ is the $n \times n$ identity matrix.  Here $\beta$ and $\sigma$ are unknown and are to be estimated in terms of $y$ and $A$.  In this paper, we will always require $p \le n$ though for certain results (indicated in the text) we will also require $p < n$.  The probability density function (PDF) of $Y$ given values of the unknown model parameters $\beta$ and $\sigma$ is therefore
\begin{eqnarray}
(2 \pi \sigma^2)^{-n/2} \exp\left( -\frac{\| y - A \beta \|^2}{2 \sigma^2} \right)
\label{E:PDFy1}
\end{eqnarray}
where $\| \cdot \|$ is the Euclidean norm on $\R^n$.

It is well-known that this statistical model is an exponential family, so we will now write (\ref{E:PDFy1}) in canonical form.  Let $B$ be any $n \times p$ matrix whose columns form an orthonormal basis for the column space $\col A$ of $A$, e.g. we could take $B = A (A^T A)^{ -\frac{1}{2}}$.  Then $B^T B = I_p$ and the orthogonal projection of $\R^n$ onto $\col A$ is $A (A^T A)^{-1} A^T = B B^T$.  Define the {\em sufficient statistics} $T(y)$ and {\em natural parameters} $\theta$ of the exponential family to be
\begin{equation}\label{E:suffstat}
T(y) \defeq \left[ \begin{array}{cc} B^T y \\
                                  \| y \|^2  \end{array} \right]
\mbox{ and } \theta \defeq \frac{1}{\sigma^2} \left[ \begin{array}{cc} B^T A \beta \\
-\frac{1}{2}  \end{array} \right].
\end{equation}
Then the PDF (\ref{E:PDFy1}) can be written in the canonical form
\begin{eqnarray}
p_Y(y | \theta) = \exp(\theta \cdot T(y)) h_Y(y)/Z(\theta)
\label{E:PDFy2}
\end{eqnarray}
where the dot denotes the Euclidean inner product, $h_Y(y) = (2\pi)^{-n/2}$ and the {\em partition function} $Z(\theta)$ is
\beq
Z(\theta) \defeq \exp \left(-\frac{n}{2} \log(-2 \theta_{p+1}) - \frac{\theta_1^2 + \ldots + \theta_p^2}{4 \theta_{p+1}} \right).
\label{E:logpartfn}
\eeq
Note from (\ref{E:suffstat}) that the natural parameter space $\Theta$, which is the set of all natural parameters, is
\beq \label{E:Theta}
\Theta = \{ \theta \in \R^{p+1} \mid  \theta_{p+1} < 0 \}.
\eeq

\begin{remark}
The first $p$ sufficient statistics $B^T y$ are essentially equal to the orthogonal projection of $y$ onto $\col A$.  More precisely, since $B^T y = B^T (B B^T y)$ and $B B^T$ is orthogonal projection, the first $p$ sufficient statistics are the orthogonal projection of $y$ onto $\col A$ written in terms of the co-ordinates for $\col A$ corresponding to the basis formed by the columns of $B$.  The reason for using this definition, instead of simply taking the orthogonal projection of $y$ onto $\col A$, is that we require the set of all possible sufficient statistics to form an open set in $\R^d$ for some $d$, while $\col A$ is a lower-dimensional set in $\R^n$.
\end{remark}

\begin{remark}
In this paper, we will think of $\Theta$ as simply being a subset of a generic $(p+1)$-dimensional vector space $\R^{p+1}$.  However, for a number of reasons, it is more natural to think of $\Theta$ as a subset of the dual space
to the vector space
containing the set $\mathcal{X}$ of all possible sufficient statistics.  One reason this is natural is that the dot in (\ref{E:PDFy2}) then becomes the natural pairing between a vector space and its dual, rather than the (non-canonical) Euclidean dot product.  Another reason is that the Fisher information matrices on $\Theta$ and $\mathcal{X}$ (when $\mathcal{X}$ is identified with the expectation parameter space, see Section \ref{SS:expparam}) are matrix inverses of each other, as is the case for a metric on a vector space and the induced metric on the dual vector space.  A third reason is that there is a close connection between exponential families and convex conjugation \cite[Chapter 9]{barndorff-nielsen} which makes it natural to think of $\Theta$ and $\mathcal{X}$ as convex subsets of dual vector spaces.  This connection
can be used to show (under mild conditions) that the maximum likelihood estimator is the gradient of the maximized log-likelihood function, and that the maximized log-likelihood function can itself be calculated as the convex conjugate of the log-partition function \cite[Theorem 9.13]{barndorff-nielsen}.
\end{remark}

\section{The linear regression model is isometric to $2n \BH^{p+1}$}
\label{S:hyp}

When equipped with the Fisher information metric, the parameter space for the linear regression model above, with $p$ covariates and a sample size of $n$, is a Riemannian manifold.  In this section, we will show that this is isometric to the Riemannian manifold $2n \BH^{p+1}$, which we define to be $(p+1)$-dimensional hyperbolic space $\BH^{p+1}$ (with all sectional curvatures equal to $-1$) equipped with a metric tensor which is $2n$ times the usual metric tensor on $\BH^{p+1}$.  This result contradicts certain findings of \cite{kass_vos} and \cite{costa} when $n \not= 1$, but we will show that the formulae of \cite{kass_vos} and \cite{costa} are not correct.

Recall that if an open set $U \subseteq \R^k$ parameterises a stochastic model then the {\em Fisher information metric} of this model is represented, in the local co-ordinates of this parameterisation, by the {\em Fisher information matrix} $g_U$. Under regularity conditions satisfied by all models considered in this paper, $g_U$ is given by either of the following expressions
\beq
g_U = \E[(\nabla \ell)(\nabla \ell)^T] = - \E[\mbox{Hess}(\ell)]
\label{E:defFI}
\eeq
where $\ell:U \to \R$ is the log-likelihood function, $\nabla \ell$ is its gradient (interpreted as a column matrix in the formula above), $\mbox{Hess}(\ell)$ is its Hessian matrix and the expectation is taken over the observed data.

\subsection{The upper half-space parameterisation}
\label{SS:upperhalfspace}

We now define a parameterisation for the linear regression model and calculate its corresponding Fisher information matrix. Let
\beq \label{E:uhparam}
\uh \defeq \left[ \begin{array}{c} B^T A \beta \\ \sigma \sqrt{2n}  \end{array} \right]
\eeq
and note that, up to a linear transformation, this is just the $\beta, \sigma$ parameterisation.  The set of possible values for $\uh$ is the upper half-space $\uhspace \defeq \{ \uh \in \R^{p+1} \mid \uh_{p+1} > 0 \}$ so, in light of this and Theorem \ref{T:upperhalfspace} below, we will refer to this as the {\em upper half-space parameterisation}.

The upper half-space model for hyperbolic space is a Riemannian manifold with a metric tensor which is a particular multiple of the identity, as given in \cite[Theorem 4.6.6]{ratcliffe}, and all sectional curvatures equal to $-1$.

\begin{theorem}
\label{T:upperhalfspace}
The Fisher information matrix for the upper half-space parameterisation is
$$g_\uhspace = 2n \uh_{p+1}^{-2} I_{p+1}$$
where $I_{p+1}$ is the $(p+1) \times (p+1)$ identity matrix.  So $\uhspace$ is the upper half-space model for $(p+1)$-dimensional hyperbolic space but with a metric tensor that is $2n$ times the usual metric tensor.
\end{theorem}

\begin{proof}
We first note that $A \beta = B \uh_{[1:p]}$, where $\uh_{[1:p]} = B^T A \beta$ is the $p \times 1$ column matrix whose entries are the first $p$ entries of $\uh$.  This follows because $B B^T$ is the identity on $\col A$ (being the orthogonal projection onto $\col A$) and $A \beta \in \col A$ so $A \beta = B B^T A \beta = B \uh_{[1:p]}$.  So by (\ref{E:PDFy1}), the log-likelihood function for this parameterisation is
$$\ell_\uhspace(\uh) = -\frac{n}{2} \log (\pi/n) - n \log \uh_{p+1} -n \uh_{p+1}^{-2} \| y - B \uh_{[1:p]} \|^2.$$
For $i,j = 1, \ldots, p$ we therefore have
$$\frac{\partial \ell_\uhspace}{\partial \uh_i} = 2n \uh_{p+1}^{-2} ( y - B \uh_{[1:p]})^T B e_i,$$
where $e_i$ is the $i^{th}$ standard basis vector for $\R^p$, and
$$\frac{\partial \ell_\uhspace}{\partial \uh_{p+1}} = - n \uh_{p+1}^{-1} + 2n \uh_{p+1}^{-3} \| y - B \uh_{[1:p]} \|^2.$$
So letting $\delta_{ij}$ be the Kronecker delta,
$$\frac{\partial^2 \ell_\uhspace}{\partial \uh_i \partial \uh_j}
= -2n \uh_{p+1}^{-2} \delta_{ij},$$
$$\frac{\partial^2 \ell_\uhspace}{\partial \uh_i \partial \uh_{p+1}} = -4n \uh_{p+1}^{-3} ( y - B \uh_{[1:p]})^T B e_i$$
and
$$\frac{\partial^2 \ell_\uhspace}{\partial \uh_{p+1}^2}
=  n \uh_{p+1}^{-2} - 6n \uh_{p+1}^{-4} \| y - B \uh_{[1:p]} \|^2.$$
Taking expectations of the negatives of these second partial derivatives and using the facts $\E[y] = A \beta = B \uh_{[1:p]}$ and
$$ \E\| y - B \uh_{[1:p]} \|^2 = \sum_{i=1}^n \E[(y_i - \E[y_i])^2] = n \sigma^2 = \frac{\uh_{p+1}^2}{2}$$
then proves $g_\uhspace = 2n \uh_{p+1}^{-2} I_{p+1}$.  Comparing this with \cite[Theorem 4.6.6]{ratcliffe} then proves the theorem.
\end{proof}

\subsection{Sectional curvatures of the linear regression model}
\label{SS:sectcurvatures}

Theorem \ref{T:upperhalfspace} allows us to see that the linear regression parameter space $\uhspace$ is a Riemannian manifold with all sectional curvatures equal to $-1/2n$.  For if $\lambda > 0$ and $(M,g)$ is a Riemannian manifold, where $M$ is a smooth manifold and $g$ is a metric tensor, then the sectional curvatures of $(M,g)$ are $\lambda^{-1}$ times the corresponding section curvatures of the Riemannian manifold $(M,\lambda g)$.
(This is elementary to prove from the relevant definitions, but as a check that the correct power of $\lambda$ here is $-1$, apply this formula to the case when $(M,g)$ is the unit $2$-sphere:  for then $(M,\lambda g)$ is isometric to the $2$-sphere with radius $\sqrt{\lambda}$ and this has all sectional curvatures equal to $\lambda^{-1}$.)
Combining this scaling result with Theorem \ref{T:upperhalfspace} and the fact that the sectional curvatures of the upper half-space model are all equal to $-1$ then proves that $\uhspace$ has all sectional curvatures equal to $-1/2n$.

\subsection{The spherical normal model}

The linear regression model can be viewed as a sub-model of the $n$-dimensional spherical normal model $y \sim N_n(\mu, \sigma^2 I_n)$ with unknown $\sigma$.
On the other hand, the spherical normal model is the special case of the linear regression model where $p=n$ and $A = B = I_n$.  Our finding from Section \ref{SS:sectcurvatures} that all linear regression models with $n$ observations and $p$ covariates have sectional curvatures of $-1/2n$ therefore contradicts Kass and Vos
\cite[\S 7.4.3]{kass_vos} when $n \not= 1$, since they report that the sectional curvatures for the model
$y \sim N_n(\beta, \sigma^2 I_n)$ are all $-1/2$ for all $n$.  However, we will now show that this result in \cite{kass_vos} cannot be correct.

Intuitively, when $n$ is large, we would expect the $n$-dimensional spherical normal model (with a fixed number $N \not= 1$ of observations) to behave like the spherical normal model with known $\sigma$.  But the $\sigma$-known model has a Euclidean geometry and hence sectional curvatures of $0$, so the sectional curvatures for the the $n$-dimensional spherical normal model should approach $0$ as $n \to \infty$.  This is consistent with our result but not with that of \cite{kass_vos}.

A more careful argument can be given by interpreting the model for $n$ independent and identically distributed univariate normal random variables $y_1, \ldots, y_n \sim N(\mu, \sigma^2)$ as a sub-model of the $n$-dimensional spherical normal model.  If we define $f(\mu, \sigma) \defeq (\mu/\sqrt{2}, \ldots, \mu/\sqrt{2}, \sigma)$ then $f$ maps the $\mu, \sigma$ parameterisation of the former model into Kass and Vos' $z$ parameterisation of the latter model (in a way that respects likelihood functions).  The Jacobian matrix $J$ of $f$ is
$$ J =  \left[ \begin{array}{cc} \vec{1}/\sqrt{2} &  \vec{0} \\
                                    0 & 1  \end{array} \right]$$
where $\vec{1}$ and $\vec{0}$
are $n \times 1$ column matrices with all entries equal to $1$ and $0$, respectively.
So by the change-of-variables formula (Lemma \ref{L:pullback}, below), if the formulae of \cite[\S 7.4.3]{kass_vos} were true then the Fisher information metric for $n$ independent and identically distributed univariate normal random variables would be
$$ J^T ( 2 \sigma^{-2} I_{n+1} ) J = 2 \sigma^{-2} J^T J =
2 \sigma^{-2} \left[ \begin{array}{cc} \vec{1}^T/\sqrt{2} &  0 \\
                                    \vec{0}^T & 1  \end{array} \right]
\left[ \begin{array}{cc} \vec{1}/\sqrt{2} &  \vec{0} \\
                                    0 & 1  \end{array} \right]
= \sigma^{-2} \left[ \begin{array}{cc} n &  0 \\
                                    0 & 2  \end{array} \right],$$
which cannot be correct because the Fisher information matrix should scale linearly with the sample size.

In a similar way, we can see that the Fisher information matrix of \cite[\S II(i)]{costa} is also incorrect.  This has been corrected in \cite{costa_2012}, though the sectional curvatures for the spherical normal model are not correct in either paper.

\section{The distribution of the sufficient statistic}
\label{S:suffstat}

If $y$ is a realization of a random variable $Y$ then the sufficient statistic $x \defeq T(y)$ is a realization of a different random variable $X = T(Y)$.  It is a remarkable fact for exponential families \cite[p. 127]{barndorff-nielsen}, provable by a direct application of the smooth co-area formula,
that the PDF $p_X(x | \theta)$ of $X$ given $\theta$ is very similar to that of $Y$, namely
\begin{eqnarray}
p_X(x | \theta) = \exp(\theta \cdot x) h_X(x)/Z(\theta)
\label{E:PDFx}
\end{eqnarray}
where $h_X(x)$ is some function of $x$ (which is not closely related to $h_Y$, in general).  Therefore the PDFs for $X$ given $\theta$ form a natural exponential family with the same natural parameter and the same partition function as the exponential family for $Y$.

Let $\mathcal{X}$ be the set of all sufficient statistics, i.e., let $\mathcal{X}$ be the image $T$ as given in (\ref{E:suffstat}).

\begin{lemma}
\label{L:calX}
When $p < n$, $\mathcal{X}$ is the solid paraboloid
\beq \label{E:calX}
\mathcal{X} = \{ x \in \R^{p+1} \mid  x_{p+1} \ge x_1^2 + \ldots + x_p^2  \}
\eeq
and when $p=n$, $\mathcal{X}$ is the paraboloid $\{ x \in \R^{p+1} \mid  x_{p+1} = x_1^2 + \ldots + x_p^2  \}$.
\end{lemma}

\begin{proof}
Deferred to the Appendix.
\end{proof}

We can now calculate the distribution of $X$ given $\theta$.  In light of (\ref{E:PDFx}), this amounts to finding $h_X(x)$, though our proof will also establish (\ref{E:PDFx}) for linear regression.

\begin{lemma} \label{L:PDFx} The PDF $p_X(x | \theta)$ of $X$ given $\theta$ is as in (\ref{E:PDFx}) where
$$ h_X(x) = c_h (x_{p+1} - x_1^2 - \ldots - x_p^2)^{\frac{n-p-2}{2}} $$
and the constant $c_h = \left( 2^{\frac{n}{2}} \pi^{p/2} \gammafn\left(\frac{n-p}{2}\right)\right)^{-1}$, with $\gammafn$ being the gamma function.
\end{lemma}

\begin{proof}
Deferred to the Appendix.
\end{proof}

\section{The SMML estimator for linear regression}
\label{S:SMML}

In this section, we first recall the definition of the SMML estimator, which is a Bayesian estimator motivated by information-theoretic considerations.  We then describe the expectation parameter space of the linear regression model and show that this can be naturally identified with the space $\mathcal{X}$ of sufficient statistics.  By Section \ref{S:hyp}, this gives $\mathcal{X}$ a hyperbolic metric, and we finish by showing that an SMML estimator corresponds to a partition of $\mathcal{X}$ into hyperbolic polytopes.

\subsection{SMML estimators}
\label{SS:SMMLdefn}

The SMML estimator with $m$ regions is defined as follows, where $m \ge 1$ is an integer \cite[Chapter 3]{wallace}.  Suppose we are given a partition $U_1, \ldots, U_m$ of $\mathcal{X}$, parameters $\theta_1, \ldots, \theta_m \in \Theta$ (the assertions) and real numbers $q_1, \ldots, q_m \in \R$ (the coding probabilities for the assertions) so that $1 = q_1+ \ldots +q_m$ and each $q_i > 0$.  Let $\htheta$ and $\hq$ be the step functions given by $\htheta(x) \defeq \theta_i$ and $\hq(x) \defeq q_i$ where $i$ is the unique integer for which $x \in U_i$.  If the data space $\mathcal{X}$ is countable then we can use this structure to transmit any data point $x \in \mathcal{X}$ to an imaginary receiver by first transmitting the assertion $\htheta(x)$ using an optimal codebook constructed from the coding probabilities $q_1, \ldots, q_m$, and second transmitting $x$ using an optimal coding based on the assertion $\htheta(x)$.  For linear regression, $\mathcal{X}$ is not countable, so we simply truncate all data points to a finite but large number $N$ of binary places and proceed as above \cite[p. 167--168]{wallace}.  Then the (idealized) length of the assertion for $x$ is $-\log \hq(x)$ and the length of the detail is $-\log p(x | \htheta(x))$, so the average length of the message used to encode $x$ is
\begin{equation}\label{E:I1}
I_1 = - \E[ \log \hq(X) + \log f(X | \htheta(X))]
\end{equation}
plus the constant $N \log 2$ \cite[p. 168]{wallace}.  Here, $X$ is a random variable distributed according to the {\em marginal PDF}
$$ r(x) \defeq \int_\Theta \pi_\Theta(\theta) p_X(x | \theta) d\theta. $$

\begin{defn}
\label{D:SMML}
An SMML estimator with $m$ regions is the function $\htheta(x)$ corresponding to any partition $U_1, \ldots, U_m$, assertions $\theta_1, \ldots, \theta_m$ and coding probabilities $q_1, \ldots, q_m$ which minimize $I_1$.
\end{defn}

Note that an SMML estimator with $m$ regions might not exist or might not be unique in general, however we will often refer to `the' SMML estimator when discussing this estimator informally.

Wallace \cite[p. 156]{wallace} gave conditions which the $U_1, \ldots, U_m$, $\theta_1, \ldots, \theta_m$ and $q_1, \ldots, q_m$ for an SMML estimator must satisfy.  In the case of an exponential family with PDF of the general form (\ref{E:PDFx}), these are
\begin{eqnarray}
U_i &=& \{ x \in \mathcal{X} \mid \lambda_i(x) \le \lambda_j(x) \mbox{ for all $j=1,\ldots, m$} \}  \label{E:U} \\
q_i &=& \int_{U_i} r(x) dx \label{E:q} \\
\theta_i &=& \fxt^{-1}\left( \frac{1}{q_i} \int_{U_i} x r(x) dx \right) \label{E:theta}
\end{eqnarray}
where $\lambda_i$ is the linear function of $x$ given by $\lambda_i(x) = -\log q_i - x \cdot \theta_i + \log Z(\theta_i)$ and $\fxt$ is an invertible function which will be defined in Section \ref{SS:expparam}, below.

Note that (\ref{E:U}) shows that each $U_i$ is a convex polytope (with respect to the affine structure on $\mathcal{X}$ inherited from its ambient vector space).  So $U_1, \ldots, U_m$ is a partition of $\mathcal{X}$ into convex polytopes.

\subsection{The expectation parameter space and its identification with the space of sufficient statistics}
\label{SS:expparam}

The expectation parameter $\xi$ corresponding to the natural parameter $\theta$ is defined to be the expected value $\E[X|\theta]$ of $X$ given $\theta$, i.e., the expected value of a random variable with the PDF $p_X(x | \theta)$ given in Lemma \ref{L:PDFx}.  Let $\Xi$ be the space of all expectation parameters and let $\fxt: \Theta \to \Xi$ be the map between the natural and expectation parameterisations, that is,
\beq \label{E:expectationmap}
\fxt(\theta) \defeq \int_\mathcal{X} x p_X(x | \theta) dx.
\eeq
Since (\ref{E:expectationmap}) expresses the expectation parameter $\xi = \fxt(\theta)$ corresponding to $\theta$ as a convex combination of elements of $\mathcal{X}$, it is clear that $\xi$ lies in the same vector space as $\mathcal{X}$.  In the case of linear regression when $p<n$, $\mathcal{X}$ is convex by (\ref{E:calX}), so (\ref{E:expectationmap}) further implies that $\xi \in \mathcal{X}$.  So in our main case of interest,
\begin{equation}\label{E:expsubsetdata}
\Xi \subseteq \mathcal{X}.
\end{equation}
In fact, it is known that the expectation parameter space $\Xi$ can be naturally identified with the interior of $\mathcal{X}$ for many exponential families \cite[Corollary 9.6]{barndorff-nielsen}.  We will sketch a proof of this fact, in the case of linear regression, after calculating the reparameterisation map $\fxt$.

By a standard result for exponential families (e.g. see \cite[Theorem 2.2.1]{kass_vos}), the partition function $Z$ is infinitely differentiable, $\fxt$ can be calculated as
\begin{eqnarray}
\fxt(\theta) = \grad |_\theta \log Z  \label{E:gradpsi}
\end{eqnarray}
(where $\grad |_\theta \log Z$ is the gradient of $\log Z$ evaluated at $\theta$) and $\fxt$ is a diffeomorphism (i.e., an infinitely differentiable function with an infinitely differentiable inverse) from $\Theta$ to $\Xi$.  So from (\ref{E:logpartfn}) and (\ref{E:gradpsi}) we have
\begin{equation}\label{E:fxt}
\fxt(\theta) = \frac{1}{-2\theta_{p+1}}\left(\theta_1, \ldots, \theta_p, n + \frac{\theta_1^2 + \ldots + \theta_p^2}{-2\theta_{p+1}} \right).
\end{equation}
It follows easily from this and the defining property of $\Xi$ (that $\Xi$ is the image of $\fxt$) that
$$\Xi = \{ \xi \in \R^{p+1} \mid \xi_{p+1} > \xi_1^2 + \ldots + \xi_p^2  \},$$
so comparing this to (\ref{E:calX}) and using (\ref{E:expsubsetdata}) shows that $\Xi$ is the interior of $\mathcal{X}$, i.e., up to a set with zero Lebesgue measure, there is a natural identification $\Xi = \mathcal{X}$ (when $p<n$).  So since $\Xi$ has a natural hyperbolic metric (by Section \ref{SS:upperhalfspace} and the fact that reparameterisation maps are isometries), this means that the interior of $\mathcal{X}$ has one, too.

Table \ref{T:reparam} gives the reparameterisation maps between the three parameterisations introduced so far.

\begin{sidewaystable}
\centering
\noindent
\begin{tabular}{|l|lll|}
\hline
  & Natural $\theta$ & Expectation $\xi$ & Upper half space $\uh$ \\
\hline
$\theta=$ & $\theta$ & $\frac{n}{V(\xi)}(\xi_1, \ldots, \xi_p,-\frac{1}{2})$ &
   $\frac{2n}{\uh_{p+1}^2}(\uh_1, \ldots, \uh_p,-\frac{1}{2})$  \\
$\xi=$ & $\frac{1}{-2\theta_{p+1}}(\theta_1, \ldots, \theta_p, n + \frac{\theta_1^2 + \ldots + \theta_p^2}{-2\theta_{p+1}})$ &
   $\xi$ & $(\uh_1, \ldots, \uh_p, \uh_1^2 + \ldots + \uh_p^2 + \frac{\uh_{p+1}^2}{2})$  \\
$\uh=$ & $\frac{1}{-2\theta_{p+1}}(\theta_1, \ldots, \theta_p,\sqrt{-4n\theta_{p+1}})$ &
   $(\xi_1, \ldots, \xi_p,\sqrt{2 V(\xi)})$ & $\uh$  \\
\hline
$\horomap(\cdot)$ & $2\theta$ & $(\xi_1, \ldots, \xi_p, \xi_{p+1} - V(\xi)/2)$ &
   $(\uh_1, \ldots, \uh_p, \uh_{p+1}/\sqrt{2})$ \\
$\horomap^{-1}(\cdot)$ & $\theta/2$ & $(\xi_1, \ldots, \xi_p, \xi_{p+1} + V(\xi))$ &
   $(\uh_1, \ldots, \uh_p, \sqrt{2} \uh_{p+1})$ \\
\hline
\end{tabular}
\caption{Maps between different parameterisations of the linear regression model, as well as some other useful quantities, where $V(\xi) = \xi_{p+1} - \xi_1^2 - \ldots -\xi_p^2$.}
\label{T:reparam}
\end{sidewaystable}

\subsection{Affine and hyperbolic lines in the expectation parameter space}
\label{SS:affinehyp}

We have just shown that the interior of the data space $\mathcal{X}$ can be naturally identified with the expectation parameter space $\Xi$.  We will now describe the relationship between the hyperbolic structure on $\Xi$ (coming from the Fisher information metric) and the affine structure on $\Xi$ (inherited from the vector space $\R^{p+1}$ containing $\Xi$).  We will show there is a natural function $\horomapx: \Xi \to \Xi$ which maps affine lines in $\Xi$ to hyperbolic lines in $\Xi$.  Since the partition $U_1, \ldots, U_m$ corresponding to an SMML estimator consists of affine convex polytopes (by Section \ref{SS:SMMLdefn}), this shows that $\horomapx(U_1), \ldots, \horomapx(U_m)$ is essentially a partition of the hyperbolic space $\mathcal{X}$ into hyperbolic convex polytopes.

Here, an {\em affine} plane $P$ is the non-empty set, in $\Xi \subseteq \R^{p+1}$, of solutions to a set of possibly non-homogeneous linear equations.  A {\em hyperbolic} plane $Q$ is any subset of the interior of $\Xi$ which contains the hyperbolic line (the image of a geodesic) through any two points of $Q$.  Note that in this terminology, affine and hyperbolic lines are just $1$-dimensional affine and hyperbolic planes (respectively).

Define $\horomapx: \Xi \to \Xi$ to be $\horomapx = \fxu \circ \horomapu \circ \fxu^{-1}$ where
$\horomapu: \uhspace \to \uhspace$ is given by
\beq
\label{E:horomap}
\horomapu(\uh) \defeq (\uh_1, \ldots, \uh_p,  \uh_{p+1}/\sqrt{2})
\eeq
and $\fxu:\uhspace \to \Xi$ is the reparameterisation map between $\uhspace$ and $\Xi$, i.e.,
\beq
\label{E:fxu}
\fxu(\uh) = (\uh_1, \ldots, \uh_p, \uh_1^2 + \ldots + \uh_p^2 + \frac{\uh_{p+1}^2}{2}),
\eeq
as can be calculated from the reparameterisation maps (\ref{E:suffstat}), (\ref{E:uhparam}) and (\ref{E:fxt}) (see Table \ref{T:reparam}).

The map $\horomapx$ can be interpreted in terms of the hyperbolic geometry as follows.  In the linear regression model, the point at infinity $\infty$ is a distinguished point on the sphere at infinity of the upper half-space $\uhspace$, and $\horomapu$ translates each point $\uh \in \uhspace$ away from $\infty$ along the geodesic through $\infty$ and $\uh$ by a distance $\log \sqrt{2}$.  Since this description only depends on the distinguished point $\infty$ and notions from hyperbolic geometry, which are both preserved by $\fxu$, the same interpretation holds for $\horomapx$.

\begin{lemma}
$P$ is an affine plane of $\Xi$ if and only if $\horomapx(P)$ is a hyperbolic plane of $\Xi$.  In particular, $\horomapx$ maps affine lines to hyperbolic lines.
\end{lemma}

\begin{proof}
The upper half-space model of hyperbolic $(p+1)$-dimensional space coincides with $\uhspace$ and the metrics on the two Riemannian manifolds are constant multiples of each other, so a hyperbolic plane of one is a hyperbolic plane of the other.  But the $p$-dimensional hyperbolic planes of the upper half space model all have a known form \cite{ratcliffe}, so the hyperbolic $p$-planes in $\uhspace$ are of the form
\beq \label{E:hypplane1}
Q = \{ \uh \in \uhspace \mid (\uh_1 - c_1)^2 + \ldots + (\uh_{p+1} - c_{p+1})^2 = R^2 \}
\eeq
or
\beq \label{E:hypplane2}
Q = \{ \uh \in \uhspace \mid c_1 \uh_1 + \ldots + c_{p+1}\uh_{p+1} = d \}
\eeq
for some $R>0$, $d \in \R$ and $c \in \R^{p+1}$ with $c_{p+1} = 0$.  And since $\fxu$ is an isometry, the hyperbolic $p$-planes in $\Xi$ are all of the form $\fxu^{-1}(Q)$ for some hyperbolic plane $Q$ in $\uhspace$.

Now, $P \subseteq \Xi$ is an affine $p$-plane if and only if $P \not= \emptyset$ and
$$ P = \{ \xi \in \Xi \mid L(\xi) = 0 \} $$
for some (affine) linear function $L:\Xi \to \R$, say $L(\xi) = a_1 \xi_1 + \ldots + a_{p+1}\xi_{p+1} + b$.  So
\begin{eqnarray}
(L \circ \fxu \circ \horomapu^{-1})(\uh)
&=& (L \circ \fxu)(\uh_1, \ldots, \uh_p,  \sqrt{2}\uh_{p+1}) \mbox{ by (\ref{E:horomap})} \nonumber \\
&=& L(\uh_1, \ldots, \uh_p, \uh_1^2 + \ldots + \uh_p^2 + \uh_{p+1}^2) \mbox{ by (\ref{E:fxu})} \nonumber \\
&=& a_1 \uh_1 + \ldots + a_p \uh_p + a_{p+1}(\uh_1^2 + \ldots + \uh_p^2 + \uh_{p+1}^2) + b \nonumber \\
&=& a_{p+1} \left( (\uh_1 - c_1)^2 + \ldots + (\uh_{p+1} - c_{p+1})^2 - R^2 \right)  \label{E:LfH} \end{eqnarray}
if $a_{p+1} \not= 0$, where $c_{p+1} = 0$, $c_i = -a_i/2a_{p+1}$ for $i=1, \ldots, p$ and $R^2 = -b/a_{p+1} + c_1^2 + \ldots + c_p^2$.  Note that $R^2 > 0$ because $P \not= \emptyset$ so $L$ has a zero in $\Xi$ and hence $L \circ \fxu \circ \horomapu$ must have a zero in $\uhspace$.  Comparing (\ref{E:LfH}) with (\ref{E:hypplane1}) when $a_{p+1} \not= 0$, or comparing a similar expression with (\ref{E:hypplane2}) when $a_{p+1} = 0$, shows that
\beq \label{E:lincomp}
\{ \uh \in \uhspace \mid (L \circ \fxu \circ \horomapu^{-1})(\uh) = 0 \} \mbox{ is a hyperbolic $p$-plane in $\uhspace$.}
\eeq

Now, if $U$ and $V$ are any two sets and $f:U \to V$ and $g:U \to \R$ are any functions with $f$ injective (one-to-one) then
$$f(\{ u \in U \mid g(u) = 0 \})
= \{ v \in V \mid g(f^{-1}(v)) = 0 \}.$$
Applying this to the case $f = \horomapu \circ \fxu^{-1}$, $g=L$, $U = \Xi$ and $V = \uhspace$ gives
\begin{eqnarray*}
\horomapx(P) &=& \horomapx(\{ \xi \in \Xi \mid L(\xi) = 0 \}) \\
&=& (\fxu \circ \horomapu \circ \fxu^{-1}) (\{ \xi \in \Xi \mid L(\xi) = 0 \}) \\
&=& \fxu(\{ \uh \in \uhspace \mid (L \circ \fxu \circ \horomapu^{-1})(\uh) = 0 \}) \\
&=& \fxu(Q)
\end{eqnarray*}
where $Q$ is a hyperbolic $p$-plane in $\uhspace$ by (\ref{E:lincomp}).  Therefore $\horomapx(P) = \fxu(Q)$ is a hyperbolic $p$-plane in $\Xi$.  Also, any hyperbolic $p$-plane arises in such a way, so this proves the lemma for $p$-dimensional affine and hyperbolic planes.  So lastly note that an affine or hyperbolic plane of any dimension can be expressed as an intersection of $p$-dimensional planes, and that such intersections always give planes, so this proves the lemma.
\end{proof}

\section{The Jeffreys prior and the marginal distribution}

In this section we will put the (improper) Jeffreys prior $\pi_\Theta(\theta)$ on $\theta$ and calculate the marginal distribution of $X$, i.e., the distribution of $X$ not conditioned on $\theta$.  We choose the Jeffreys prior because it is natural, it makes few assumptions about the parameter values (i.e., it is uninformative) and it is tractable to work with.  It also has a geometrical interpretation, so this choice preserves the symmetries of, and hence the close connections with, the underlying hyperbolic geometry.

\subsection{The Jeffreys prior on the natural parameter space}

From the definition (\ref{E:defFI}) and the expression (\ref{E:PDFy2}), it is easy to see that the Fisher information matrix $g_\Theta$ corresponding to the natural parameterisation of the linear regression model (or any other exponential family \cite{kass_vos}) is the Hessian of the log-partition function.  So from (\ref{E:logpartfn}),
\beq
\label{E:FImatrixTheta}
g_\Theta =  \frac{1}{-2\theta_{p+1}}  \left[ \begin{array}{cc} I_p &  -\theta_{p+1}^{-1} \theta_{[1:p]}  \\
-\theta_{p+1}^{-1} \theta_{[1:p]}^T &  -n \theta_{p+1}^{-1} + \theta_{p+1}^{-2}(\theta_1^2 + \ldots + \theta_p^2) \end{array} \right]
\eeq
where $\theta_{[1:p]}$ is the $p \times 1$ column matrix with entries $\theta_1, \ldots, \theta_p$.  Recall that $\theta_{p+1}<0$ so all entries of $g_\Theta$ are positive.  The (improper) Jeffreys prior is defined to be $\pi_\Theta(\theta) \defeq \sqrt{\det g_\Theta}$, so using (\ref{E:FImatrixTheta}) and expanding the determinant of $g_\Theta$ along its bottom row gives
\beq \label{E:JeffreysTheta}
\pi_\Theta(\theta) = \sqrt{n} 2^{-\frac{p+1}{2}} (-\theta_{p+1})^{-\frac{p+2}{2}}.
\eeq

\subsection{The marginal distribution}
\label{SS:marginal}

We can now calculate the marginal distribution of $X$ (not conditioned on $\theta$), whose PDF $r(x)$ is defined to be
$$ r(x) \defeq \int_\Theta \pi_\Theta(\theta) p_X(x | \theta) d\theta. $$

\begin{lemma}
\label{L:marginal}
If $\pi_\Theta(\theta)$ is the Jeffreys prior then the marginal distribution of $X$ is
$$ r(x) = c_r (x_{p+1} - x_1^2 - \ldots - x_p^2)^{-\frac{p+2}{2}}$$
where $c_r = \sqrt{n} 2^{\frac{p-1}{2}} \gammafn\left( \frac{n}{2} \right) / \gammafn\left(\frac{n-p}{2}\right)$ and
$\gammafn$ is the gamma function.
\end{lemma}

\begin{proof}
Deferred to the Appendix.
\end{proof}

\subsection{The marginal PDF is a multiple of the hyperbolic volume density}
\label{SS:marginal_hypvol}

We will now show that the marginal PDF on $\mathcal{X}$ corresponding to the Jeffreys prior is a constant multiple of the hyperbolic volume density (recall that the interior of $\mathcal{X}$ has a natural hyperbolic metric by Section \ref{SS:expparam}).

Recall that $\Theta$ and $\Xi$ are the natural and expectation parameterisations of the linear regression model, that their Fisher information matrices are $g_\Theta$ and $g_\Xi$ (respectively) and that the reparameterisation map $f_{\Xi\Theta} : \Theta \to \Xi$ between them is given by (\ref{E:fxt}).  Let $\pi_\Xi(\xi) = \sqrt{\det g_\Xi}$ be the volume density (i.e., the improper Jeffreys prior) on $\Xi$.  Then a standard result for exponential families \cite[Theorem 2.2.5]{kass_vos} is that $g_\Xi(\xi) = g_\Theta^{-1}(\fxt^{-1}(\xi))$ where $g_\Theta^{-1}$ is the matrix inverse of $g_\Theta$ and $\fxt^{-1}$
is the inverse function of $\fxt$.  Therefore $\pi_\Xi(\xi) = \left(\det g_\Theta^{-1}(\fxt^{-1}(\xi))\right)^{\frac{1}{2}} = \left(\pi_\Theta(\fxt^{-1}(\xi))\right)^{-1}$.
It is easy to show from (\ref{E:fxt}) that
$$\fxt^{-1}(\xi) = \frac{n}{\xi_{p+1} - \xi_1^2 - \ldots -\xi_p^2}\left(\xi_1, \ldots, \xi_p,-\frac{1}{2}\right)$$
so, from (\ref{E:JeffreysTheta}),
\begin{eqnarray*}
\pi_\Xi(\xi)
&=& n^{-\frac{1}{2}} 2^{\frac{p+1}{2}} \left(\frac{n}{2(\xi_{p+1} - \xi_1^2 - \ldots -\xi_p^2)} \right)^{\frac{p+2}{2}} \\
&=& n^{\frac{p+1}{2}} 2^{-\frac{1}{2}} (\xi_{p+1} - \xi_1^2 - \ldots -\xi_p^2)^{-\frac{p+2}{2}} \\
&=& \left( \frac{\gammafn\left(\frac{n-p}{2}\right)}{\gammafn\left( \frac{n}{2} \right)} \left(\frac{n}{2}\right)^{\frac{p}{2}} \right) r(\xi)   \mbox{ by Lemma \ref{L:PDFx}.}
\end{eqnarray*}
Up to a constant factor, the marginal probability $r(x)$ is therefore the hyperbolic volume density.  Furthermore, it is not hard to see that the factor is approximately $1$ when $p \ll n$ and $p$ and $n$ are even.

\appendix
\section{Proofs of technical lemmas}
\label{S:proofs}

We begin with a lemma which shows that Fisher information matrices behave well under reparameterisations, inclusions and submersions.  In particular, this will show that the Fisher information matrices determine a well-defined metric on the underlying stochastic manifold (though it is not hard to prove this fact directly by giving a coordinate-free definition of the metric).

Let $U$ and $V$ be parameter spaces (of arbitrary dimensions) for two stochastic models and let $\ell_U: U \to \R$ and $\ell_V: V \to \R$ be the corresponding log-likelihood functions.  If $f:U \to V$ is a function so that $\ell_U = \ell_V \circ f$ then we say that {\em $f$ maps $U$ into $V$ as a parameterised sub-model}.
\begin{lemma}
\label{L:pullback}
If $f:U \to V$ is a differentiable map which maps $U$ into $V$ as a parameterised sub-model then
$$g_U = J^T g_V J$$
where $g_U$ and $g_V$ are the Fisher information matrices of the two parameterisations and $J$ is the Jacobian matrix of $f$.  In other words, $g_U$ is the pull-back of $g_V$ via $f$.
\end{lemma}

\begin{proof}
By definition, $g_U = \E[(\nabla \ell_U)(\nabla \ell_U)^{T}]$ and $g_V = \E[(\nabla \ell_V)(\nabla \ell_V)^{T}]$.  By the chain rule, $\nabla \ell_U = J^T \nabla \ell_V$, where $\nabla \ell_U$ and $\nabla \ell_V$ are gradients of $\ell_U$ and $\ell_V$.  Therefore
$$ g_U = \E[(\nabla \ell_U)(\nabla \ell_U)^{T}] = \E[J^T (\nabla \ell_V)(\nabla \ell_V)^{T} J]
= J^T \E[(\nabla \ell_V)(\nabla \ell_V)^{T}] J = J^T g_V J,$$
as required.
\end{proof}

We now give proofs for some technical lemmas.

\begin{proof}[Proof of Lemma \ref{L:calX}]
Recall that $B B^T$ is the orthogonal projection onto $\col A$, so $1 - B B^T$ is the orthogonal projection onto the space perpendicular to $\col A$ and hence
$$\| y \|^2 = \| B B^T y \|^2 + \| (1 - B B^T) y \|^2$$
by Pythagoras' theorem.  Since the columns of $B$ form an orthonormal basis for $\col A$, $\| B B^T y \|^2 = \| B^T y \|^2_{\R^p}$, where the second norm is the Euclidean norm on $\R^p$ (and, as above, the norm without a subscript is the Euclidean norm on $\R^n$).  So substituting $\| B B^T y \|^2 = \| B^T y \|^2_{\R^p} = \| (x_1, \ldots, x_p) \|^2_{\R^p} = x_1^2 + \ldots + x_p^2$ and $\| y \|^2 = x_{p+1}$ into the above formula we obtain
\beq \label{E:parab}
x_{p+1} = x_1^2 + \ldots + x_p^2 + \| (1 - B B^T) y \|^2.
\eeq
Since $\| (1 - B B^T) y \|^2 \ge 0$, (\ref{E:parab}) implies $x_{p+1} \ge x_1^2 + \ldots + x_p^2$ and hence that the image of $T$ lies in $\mathcal{X}$.

On the other hand, if $p<n$ then there exists a non-zero vector $v$ perpendicular to $\col A$, so given any
$x \in \mathcal{X}$, if we define $y = B x_{[1:p]} + tv$ where $x_{[1:p]}$ is the $p \times 1$ column matrix with entries $x_1, \ldots, x_p$ and $t = \sqrt{x_{p+1} - x_1^2 - \ldots - x_p^2}$ then $T(y) = x$, so the image of $T$ also contains $\mathcal{X}$.  Here, $T(y) = x$ follows by using $B^T v = 0$ and $B^T B = I_p$ to show that $B^T y = B^T (B x_{[1:p]} + tv) = x_{[1:p]}$ so $\| y \|^2 = \| B B^T y \|^2 + \| (1 - B B^T) y \|^2 = \| B x_{[1:p]} \|^2 + \| (1 - B B^T) y \|^2 = \| x_{[1:p]} \|_{\R^p}^2 + \| (1 - B B^T) y \|^2 = x_1^2 + \ldots + x_p^2 + t^2 \| v \|^2$.
\end{proof}

\begin{proof}[Proof of Lemma \ref{L:PDFx}]
Let $x = T(y)$ be the sufficient statistic and let $x_{[1:p]}$ be the $p \times 1$ column matrix whose entries are the first $p$ sufficient statistics, so $x_{[1:p]} = B^T y$ by (\ref{E:suffstat}).  Then since $y$ given $\beta$ and $\sigma$ is normally distributed, so is $x_{[1:p]}$.  Also, the expected value of $x_{[1:p]}$ is $B^T \E[y] = B^T A \beta$ and the variance-covariance matrix of $x_{[1:p]}$ is
$$B^T \var(y) B = B^T (\sigma^2 I_n) B = \sigma^2 I_p.$$
So $x_{[1:p]} \sim N_p(B^TA\beta, \sigma^2 I_p)$ and the PDF of $x_{[1:p]}$ given $\theta$ is
\begin{equation}\label{E:jointPDF1top}
p(x_1, \ldots, x_p | \theta) = (2 \pi \sigma^2)^{-p/2} \exp\left( -\frac{\| x_{[1:p]} - B^T A \beta \|^2_{\R^p}}{2 \sigma^2} \right)
\end{equation}
where $\| \cdot \|^2_{\R^p}$ is the Euclidean norm on $\R^p$ (and recall that the norm $\| \cdot \|^2$ without a subscript is the Euclidean norm on $\R^n$).

Now, from (\ref{E:parab}) and an equation immediately preceding it, we have
$$x_{p+1} = x_1^2 + \ldots + x_p^2 + \| (1 - B B^T) y \|^2$$
and $x_1^2 + \ldots + x_p^2 = \| B B^T y \|^2$.  But $y$ is a normal random variable and $B B^T y$ and $(1 - B B^T) y$ are uncorrelated, hence they are independent and so are their norms $x_1^2 + \ldots + x_p^2$ and $\| (1 - B B^T) y \|^2$.  Therefore
$$x_{p+1} = x_1^2 + \ldots + x_p^2 + \sigma^2 Q$$
where $Q$ is a chi-squared random variable with $n-p$ degrees of freedom which is independent of $x_1, \ldots, x_p$.  So $x_{p+1}$ given $x_1, \ldots, x_p$ and $\theta$ is a deterministic linear function of $Q$, hence its PDF $p(x_{p+1} | x_1, \ldots, x_p, \theta)$ can be calculated from the PDF of $Q$
and the change of variables formula for PDFs
as
\beq
\frac{1}{\sigma^2 2^{\frac{n-p}{2}} \gammafn\left(\frac{n-p}{2}\right)}
\left( \frac{x_{p+1}-x_1^2 - \ldots - x_p^2}{\sigma^2}  \right)^{\frac{n-p-2}{2}}
\exp\left( -\frac{x_{p+1}-x_1^2 - \ldots - x_p^2}{2 \sigma^2} \right).
\label{E:jointPDF1last}
\eeq
Combining (\ref{E:jointPDF1top}) and (\ref{E:jointPDF1last}) then gives the PDF of $X$ given $\theta$:
\begin{eqnarray*}
p_X(x | \theta) &=& p(x_{p+1} | x_1, \ldots, x_p, \theta) \, p(x_1, \ldots, x_p | \theta)  \\
&=& \sigma^{-n}  \exp\left( \frac{ x_{p+1}-x_1^2 - \ldots - x_p^2 + \| x_{[1:p]} - B^T A \beta \|^2_{\R^p}}{-2 \sigma^2} \right) \\
&& \times \left( 2^{\frac{n}{2}} \pi^{p/2} \gammafn\left(\frac{n-p}{2}\right)\right)^{-1}
\left( x_{p+1}-x_1^2 - \ldots - x_p^2 \right)^{\frac{n-p-2}{2}} \\
&=& \sigma^{-n} \exp\left( \frac{ x_{p+1} -2 x_{[1:p]}\cdot B^T A \beta + \| B^T A \beta \|^2_{\R^p}}{-2 \sigma^2} \right) h_X(x) \\
&=& \exp(\theta \cdot x) h_X(x) / Z(\theta)
\end{eqnarray*}
by (\ref{E:suffstat}) and (\ref{E:logpartfn}).
\end{proof}

\begin{proof}[Proof of Lemma \ref{L:marginal}]
From (\ref{E:JeffreysTheta}) and Lemma \ref{L:PDFx},
$$ r(x) =  \sqrt{n} 2^{-\frac{p+1}{2}} h_X(x) \int_\Theta (-\theta_{p+1})^{-\frac{p+2}{2}} \exp(\theta \cdot x) \frac{1}{Z(\theta)} d\theta.$$
But from (\ref{E:logpartfn}),
\begin{eqnarray*}
\exp(\theta \cdot x)/Z(\theta)
&=& (-2 \theta_{p+1})^{\frac{n}{2}} e^{\theta_{p+1}x_{p+1}}  \exp\left(\theta_1 x_1+ \ldots + \theta_p x_p + \frac{\theta_1^2 + \ldots + \theta_p^2}{4 \theta_{p+1}}\right) \\
&=& (-2 \theta_{p+1})^{\frac{n}{2}} e^{\theta_{p+1}x_{p+1}}  \exp\left(\frac{1}{4 \theta_{p+1}} \sum_{i=1}^p \left[ 4 \theta_{p+1} \theta_i x_i  + \theta_i^2\right]\right) \\
&=& (-2 \theta_{p+1})^{\frac{n}{2}} e^{\theta_{p+1}x_{p+1}}  \exp\left(\frac{1}{4 \theta_{p+1}} \sum_{i=1}^p \left[
(\theta_i + 2 \theta_{p+1} x_i  )^2 - 4 \theta_{p+1}^2 x_i^2 \right]\right) \\
&=& (-2 \theta_{p+1})^{\frac{n}{2}} e^{\theta_{p+1}\left(x_{p+1} - x_1^2 - \ldots - x_p^2\right)}  \exp\left(\frac{1}{4 \theta_{p+1}} \sum_{i=1}^p (\theta_i + 2 \theta_{p+1} x_i )^2\right) \\
&=& (-2 \theta_{p+1})^{\frac{n}{2}} e^{\theta_{p+1}\left(x_{p+1} - x_1^2 - \ldots - x_p^2\right)}
\exp\left(\frac{1}{4 \theta_{p+1}}  \| \theta_{[1:p]} + 2 \theta_{p+1} x_{[1:p]} \|^2_{\R^p} \right) \\
&=& (-2 \theta_{p+1})^{\frac{n}{2}} e^{\theta_{p+1}\left(x_{p+1} - x_1^2 - \ldots - x_p^2\right)}
(-4\pi \theta_{p+1})^{\frac{p}{2}} f(\theta_{[1:p]}) \\
&=& 2^{\frac{n}{2} + p} \pi^{\frac{p}{2}} (-\theta_{p+1})^{\frac{n+p}{2}} e^{\theta_{p+1}\left(x_{p+1} - x_1^2 - \ldots - x_p^2\right)} f(\theta_{[1:p]})
\end{eqnarray*}
where $f(\theta_{[1:p]})$ is the PDF for a normal random variable $N_p(-2 \theta_{p+1} x_{[1:p]}, -2 \theta_{p+1} I_p)$ evaluated at $\theta_{[1:p]}$.  Therefore
\begin{eqnarray*}
r(x) &=&  \sqrt{n} 2^{\frac{n+p-1}{2}} \pi^{\frac{p}{2}} h_X(x) \int_\Theta (-\theta_{p+1})^{\frac{n-2}{2}}  e^{\theta_{p+1}\left(x_{p+1} - x_1^2 - \ldots - x_p^2\right)} f(\theta_{[1:p]}) d\theta \\
&=&  \sqrt{n} 2^{\frac{n+p-1}{2}} \pi^{\frac{p}{2}} h_X(x) \int_{-\infty}^0 (-\theta_{p+1})^{\frac{n-2}{2}}  e^{\theta_{p+1}\left(x_{p+1} - x_1^2 - \ldots - x_p^2\right)} d\theta_{p+1} \mbox{ by (\ref{E:Theta})}\\
&=&  \sqrt{n} 2^{\frac{n+p-1}{2}} \pi^{\frac{p}{2}} h_X(x) \int_0^\infty e^{-st}  t^{\frac{n-2}{2}}  dt
\end{eqnarray*}
where $t = -\theta_{p+1}$ and $s = x_{p+1} - x_1^2 - \ldots - x_p^2$.  But the Laplace transform of $t^{\frac{n-2}{2}}$ is $\gammafn\left( \frac{n}{2} \right) s^{-\frac{n}{2}}$,
so
\begin{eqnarray*}
r(x) &=&  \sqrt{n} 2^{\frac{n+p-1}{2}} \pi^{\frac{p}{2}} h_X(x) \gammafn\left( \frac{n}{2} \right) (x_{p+1} - x_1^2 - \ldots - x_p^2)^{-\frac{n}{2}}  \\
&=&  \sqrt{n} 2^{\frac{n+p-1}{2}} \pi^{\frac{p}{2}}  \left( 2^{\frac{n}{2}} \pi^{p/2}
\gammafn\left(\frac{n-p}{2}\right)\right)^{-1} \gammafn\left( \frac{n}{2} \right)
(x_{p+1} - x_1^2 - \ldots - x_p^2)^{\frac{-n + n-p-2}{2}} \\
&=& c_r (x_{p+1} - x_1^2 - \ldots - x_p^2)^{-\frac{p+2}{2}}.
\end{eqnarray*}
where we have again used Lemma \ref{L:PDFx}.
\end{proof}

\end{document}